\documentclass{article}
\title{Exponential-Time Approximation (Schemes) for Vertex-Ordering Problems}
\usepackage[numbers]{natbib}
\usepackage{a4wide}

\date{}

\author{Matthias Bentert \and Fedor V. Fomin \and Tanmay Inamdar \and Saket Saurabh}

\usepackage{mathdots}
\usepackage{amsthm}
\usepackage{amsmath}
\usepackage{graphicx}
\usepackage{nicefrac}
\usepackage{tikz}
\usepackage{todonotes}
\setuptodonotes{inline}
\usepackage{tabularx}
\usepackage[boxed,noline,noend,ruled,linesnumbered]{algorithm2e}
\usepackage{dsfont}
\usepackage{xcolor}

\usepackage{amssymb}
\usepackage[capitalize,noabbrev]{cleveref}

\newtheorem{theorem}{Theorem}
\newtheorem{proposition}{Proposition}
\newtheorem{lemma}{Lemma}

\crefname{algocf}{Algorithm}{Algorithms}

\newcommand{\problemdef}[3]{
    \begin{quote}
      \normalsize\textsc{#1} \smallskip \\
      \begin{tabularx}{0.9\textwidth}{@{}l@{\hspace{3pt}}X}
        \normalsize\textbf{Input:}    & \normalsize#2 \\
        \normalsize\textbf{Task:} & \normalsize#3
      \end{tabularx}
    \end{quote}
}
\usepackage{graphicx}

\newcommand{\fas}{\textsc{Feedback Arc Set}}
\newcommand{\wfas}{\textsc{Weighted Feedback Arc Set}}
\newcommand{\dkmc}{\textsc{Directed Minimum $(k,n-k)$-Cut}}
\newcommand{\wdkmc}{\textsc{Weighted Directed Minimum $(k,n-k)$-Cut}}
\newcommand{\ola}{\textsc{Optimal Linear Arrangement}}

\newcommand{\dola}{\textsc{Directed Optimal Linear Arrangement}}
\newcommand{\wcw}{\textsc{Weighted Directed Cutwidth}}

\DeclareMathOperator{\poly}{poly}
\DeclareMathOperator{\opt}{OPT}
\DeclareMathOperator{\cut}{cut}
\newcommand{\wmax}{{\ensuremath{w_{\max}}}}

\newcommand{\N}{\ensuremath{\mathds{N}}}

\begin{document}

\maketitle

\begin{abstract}
In this paper, we begin the exploration of vertex-ordering problems through the lens of exponential-time approximation algorithms.
In particular, we ask the following question:
    Can we simultaneously beat the running times of the fastest known (exponential-time) exact algorithms and the best known approximation factors that can be achieved in polynomial time? Following the recent research initiated by Esmer et al. (ESA~2022, IPEC~2023, SODA~2024) on vertex-subset problems, and by Inamdar et al. (ITCS~2024) on graph-partitioning problems, we focus on vertex-ordering problems. 
    In particular, we give positive results for \textsc{Feedback Arc Set}, \textsc{Optimal Linear Arrangement}, \textsc{Cutwidth}, and \textsc{Pathwidth}.
    Most of our algorithms build upon a novel ``balanced-cut'' approach---which is our main conceptual contribution.
    This allows us to solve various problems in very general settings allowing for directed and arc-weighted input graphs.
    Our main technical contribution is a $(1+\varepsilon)$-approximation for any~$\varepsilon > 0$ for (weighted) \textsc{Feedback Arc Set} in~${O^*((2-\delta_\varepsilon)^n)}$ time, where~$\delta_\varepsilon > 0$ is a constant only depending on~$\varepsilon$.
\end{abstract}
\section{Introduction}

Many NP-hard problems admit algorithms running in time $O^*(2^n)$\footnote{Throughout the introduction, $n$ will have the ``natural'' meaning, e.g., number of variables/vertices/elements, and the $O^*$-notation hides polynomial factors. See \cref{sec:prelim} for more details about the latter.}. 
Moreover, for fundamental problems like \textsc{Satisfiability} or \textsc{Set Cover}, there are plausible conjectures, namely Strong Exponential Time Hypothesis (SETH) or the Set Cover Conjecture, respectively, which state that these problems do not admit algorithms running in time $O^*((2-\varepsilon)^n)$ for any~${\varepsilon > 0}$. Although these conjectures have been used as a source of analogous running-time lower bound for other problems, many graph problems are not covered by such results. Indeed, quite a few long-standing exponential-time barriers have finally been overcome in recent years~\cite{Nederlof20,Zamir21,NPSW23}, giving new life to the field of \emph{moderately exponential-time algorithms}~\cite{Fomin:2010mo}.

Improving upon the trivial~$O^*(2^n)$-time algorithm for ``vertex-subset problems'' on graphs, Fomin et al.~\cite{FominGLS16} designed a general framework called \emph{monotone local search} (MLS). This framework enables the creation of state-of-the-art exact exponential algorithms for a wide range of such problems by leveraging the corresponding parameterized algorithms as black-box components. 
On the other hand, for graph partitioning and vertex-ordering problems, even the~$O^*(2^n)$-time algorithm does not follow from a trivial brute force of the search space. Nevertheless, many such problems do admit $O^*(2^n)$-time algorithms via advanced techniques such as Held-Karp-style dynamic programming~\cite{HK61,BFK+12} or speeding up the natural dynamic program via fast subset convolution~\cite{Tanmay24}. For these latter two classes of problems, beating the~$O^*(2^n)$ barrier has proven to be challenging, with the notable exception of \textsc{$k$-Cut} for which Lokshtanov et al.~\cite{Lokshtanov0S23} designed an~${O^*((2-\varepsilon)^n)}$-time algorithm for some fixed~$\varepsilon > 0$. 
Due to the lack of any progress for many of these intensively studied vertex-ordering problems, we think it is the right time to relax the question just slightly in the following way.
\begin{quote}
    Can we show or exclude $(1+\varepsilon)$-approximations for vertex-ordering problems that run in~$O((2-\delta)^n)$ time?
\end{quote}
We show that this question can be answered affirmatively for \fas, but we were not able to achieve the same for any other problem.
However, the methods we develop for \fas{} allow us to answer another natural (weaker) relaxation:
\begin{quote}
    Can we simultaneously beat the running times of the fastest known (exponential-time) exact algorithms and the best known approximation factors that can be achieved in polynomial-time for vertex-ordering problems?
\end{quote}
Given the difficulty in surmounting the current-best exact exponential-time algorithms, there has been a growing research interest in answering the latter relaxation.
This is also justified given the fact that most of the problems are known to be APX-hard---they do not admit arbitrarily good approximations. In some cases, problems do not even admit constant-factor approximations in polynomial time assuming~$P \neq NP$.
To the best of our knowledge, the first exponential-time approximation algorithms were given by Williams~\cite{Williams04} who showed how to generalize~$O^*(2^{\omega n/3})$-time algorithms for (unweighted) \textsc{Max Cut} and \textsc{Max 2-SAT} to~$(1+\varepsilon)$-approximations for the weighted variants using ``rounding tricks''; here~$\omega$ denotes the fast matrix multiplication constant.
Quite recently, Alman et al.~\cite{AlmanCW20} showed that \textsc{Max Sat} can also be approximated faster than~$O^*(2^n)$, that is, they showed that for each~$\varepsilon>0$, there is a~$(1+\varepsilon)$-approximation in~$O^*((2-\delta_\varepsilon)^n)$ time where~$\delta_\varepsilon > 0$ is a constant only depending on~$\varepsilon$.
Moreover, the monotone local search framework of Fomin et al.~\cite{FominGLS16} was extended to give exponential-time approximation schemes for vertex-subset problems in the works of Esmer et al.~\cite{EsmerKMNS22ESA,EsmerKMNS23IPEC,EKMNS24}. Even more recently, Inamdar et al.~\cite{Tanmay24} gave a framework for designing exponential-time approximation schemes for graph partitioning (i.e., cut) problems and a few more examples for specific problems are known~\cite{BEP11,CLN13,BCLNN19}.
However, none of these techniques seem applicable to vertex-ordering problems, where the objective is to find an ordering of vertices that satisfies some polynomial-time satisfiable predicate, such as \textsc{Feedback Arc Set, Optimal Linear Arrangement, Cutwidth}, and \textsc{Pathwidth}.
For all of the above except \textsc{Pathwidth}, the $O^*(2^n)$-time dynamic programming remains the best known algorithm--even including constant-factor approximations.
Moreover, the question whether e.g., \textsc{Feedback Arc Set} or \textsc{Cutwidth} can be solved faster than~$O^*(2^n)$ have been explicit open problems for a long time~\cite{Fomin:2010mo,HKPS10,HPSW13,CLPPS14}.
This reinforces our belief that it is time to study exponential-time approximation algorithms for vertex-ordering problems.
The only example of such a study is in the case of (directed or undirected) bandwidth~\cite{CP10,JKLSS19}.
Our results are summarized in Table \ref{tab:results}.
We believe that working on these vertex-ordering problems from a new angle will further enhance our understanding of them and we hope that there might be some interesting connections to find that use ideas from the vast knowledge acquired on these problems both from the world of polynomial-time approximations and from the world of exact exponential-time algorithms.

\begin{table}[]
    \centering
    \begin{tabular}{l c l}
        problem & approximation factor & running time \\\hline
        \textsc{Feedback Arc Set} & $1+\varepsilon$ & $O^*((2-\delta)^n)$ \\
        \textsc{Directed Cutwidth} & $2$ & $O^*(2^{\frac{\omega n}{3}})$ \\
        \ola & $1.5 + \varepsilon$ & $O^*((2-\delta)^n)$\\
        \dola & $2+\varepsilon$ & $O^*((2-\delta)^n)$\\
        \textsc{Directed Pathwidth} & $2$ & $O^*(1.66^n)$
    \end{tabular}
    \caption{An overview of our results for unweighted input graphs. For each problem in the left column, we show an algorithm with approximation factor in the middle column and the running time is shown in the right column. The value~$\omega < 2.373$ is the fast matrix multiplication exponent and~$\delta > 0$ is a constant that depends on~$\varepsilon$.}
    \label{tab:results}
\end{table}
In the rest of this section, we survey a few different vertex-ordering problems and previous work related to them.
Afterwards, we give an overview of our results and a high-level intuition of our main conceptual contribution.

\paragraph*{Problem Definitions and Related Work.}
We now introduce the different problems used in this paper and briefly state relevant related work for each of them.

We start with a directed version of \textsc{Minimum Cut} with an additional size constraint.

\problemdef{\dkmc}
{A directed graph~$G$ and an integer~$k$.}
{Find a set~$L$ of exactly~$k$ vertices such that the number of arcs from~$V \setminus L$ to~$L$ is minimized.}

The undirected version of this problem was introduced by Bonnet et al.~\cite{BEPT15}, who showed fixed-parameter tractability of the problem when parameterized by~$k+\opt$, where $\opt$ is the minimum number of arcs from~$V \setminus L$ to~$L$ for any set~$L$ of size~$k$.
We are not aware of any mentioning of the directed problem anywhere.
We also study a weighted generalization where we are given an additional non-negative arc-weight function~$w$ and instead of minimizing the number of arcs from~$V \setminus L$ to~$L$, we instead minimize the weight of the arc set from~$V \setminus L$ to~$L$.
\dkmc{} will be at the heart of most of our algorithms.

We continue with \textsc{Feedback Arc Set}.
It can be defined both in terms of a subset of arcs of an input graph and as a vertex-ordering problem.
The former definition asks for a given directed graph~$G=(V,A)$ to find a minimum-size set of arcs whose removal turns~$G$ into a directed acyclic graph (DAG).
We give the definition as a vertex-ordering problem as this will turn out to be more useful for us later.

\problemdef{\textsc{Feedback Arc Set}}
{A directed graph~$G=(V,A)$.}
{Find an ordering~$\pi$ of~$V$ minimizing the number of backward arcs with respect to~$\pi$.}

\textsc{Feedback Arc Set} can be solved in~$O(2^nn^2)$ time by the Held-Karp algorithm~\cite{HK61}.
The weighted version of this problem asks for an ordering such that the sum of weights of all backward arcs is minimized.
The best known polynomial-time approximaiton algorithm (for both the weighted and unweighted variants) achieves an approximation factor of~$O(\log(n) \log\log(n))$~\cite{ENSS98} and any approximation better than~$1.36$ in polynomial time would refute~$P \neq NP$~\cite{Kann92,DS05}.
As stated above, we achieve a~$(1+\varepsilon)$-approximation for \textsc{Feedback Arc Set} faster than~$O(2^n)$ for any~$\varepsilon>0$.
We mention that some of the methods used for this result can also be applied to other vertex-ordering problems (to obtain constant-factor approximations).
We list a few such problems in the following.

In \textsc{Optimal Linear Arrangement} (also known as \textsc{Minimum Linear Arrangement}), one is asked for an ordering of a graph that minimizes the sum of stretches of all (backward) arcs.
Equivalently, the problem asks for an ordering minimizing the sum of cut sizes over all positions.
In the following, we denote by~$\cut_{\pi}^i$ for an ordering~$\pi$ of the vertices and an integer~$i$, the set of all edges (or arcs) from a vertex after position~$i$ to a vertex at position at most~$i$.

\problemdef{\textsc{Optimal Linear Arrangement}}
{An directed or undirected graph~$G$ with vertex set~$V$.}
{Find an ordering~$\pi$ of~$V$ minimizing~$\sum_{i=1}^{n-1}|\cut_{\pi}^i|$.}

In the weighted setting, we want to minimize~$\sum_{i=1}^{n-1}\sum_{e \in \cut_{\pi}^i}w(e)$.
\ola{} can be solved in~$O^*(2^n)$ time~\cite{BFK+12} and the best known polynomial-time algorithm achieves an approximation factor in~$O(\sqrt{\log(n)}\log\log(n))$~\cite{FL07,CHKR10}.
It remains NP-hard even when restricted to interval graphs but there it can be 2-approximated in polynomial time~\cite{CFHKK06}.

In the problem \textsc{Directed Cutwidth}\footnote{We mention in passing that there is some inconsistency in the existing literature regarding the name \textsc{Directed Cutwidth}. An alternative definition takes a directed acyclic graph as input and searches for a topological ordering of the vertices minimizing the (undirected) cutwidth.}, we are looking for an ordering that minimizes the maximum size of any~$\cut_{\pi}^i$.
Formally, it is defined as follows.

\problemdef{\textsc{Directed Cutwidth}}
{A directed graph~$G = (V,A)$.}
{Find an ordering~$\pi$ of~$V$ minimizing~$\max_{i \in [n-1]}|\cut_{\pi}^i|$.}

In the weighted setting, we want to minimize the weight~$\max_{i\in [n-1]}\sum_{e \in \cut_{\pi}^i}w(e)$ instead of the cardinality.
\textsc{Directed Cutwidth} can be solved in~$O(2^nn^2)$ time and the weighted version can be solved in~$O(2^nn^2\log(\wmax))$ time by the Held-Karp algorithm~\cite{HK61}.
The best known approximation factor achieved by a polynomial-time algorithm is~$O(\log^{1.5}(n))$~\cite{APW12}.
A quantum algorithm running in~$O(1.817^n)$ time is also known~\cite{ABIKPV19} and if the input graph is undirected, then the problem admits a~$\log^{1+o(1)}(n)$-approximation~\cite{BKS24}.

Finally, we study \textsc{Directed Pathwidth}.
It is a directed generalization of \textsc{Pathwidth} which itself is variation of treewidth where the tree decomposition is restricted to being a path.
It is known that \textsc{Pathwidth} has an equivalent definition in terms of vertex orderings and this remains true also for the directed variant~\cite{KKKTT16}.
For a position~$i$, consider the number of vertices~$v$ that appear before position~$i$ that have in-neighbors that appear after position~$i$.
The equivalent definition of \textsc{Directed Pathwidth} asks for an ordering that minimizes the maximum quantity over all positions.
The formal definition in terms of vertex orderings is as follows.

\problemdef{\textsc{Directed Pathwidth}}
{A directed graph~$G=(V,A)$.}
{Find an ordering~$\pi$ of~$V$ minimizing $$\max_{i \in [n-1]} |\{v \in V \mid \pi(v) \leq i \land (\exists u \in V.\ \pi(u) > i \land (u,v) \in A) \}|.$$}
The fastest exact algorithm for \textsc{Directed Pathwidth} runs in~$O(1.89^n)$ time~\cite{KKKTT16} and the best known approximation factor of a polynomial-time algorithm is~$O(\log^{1.5}(n))$~\cite{FHL08}.
The undirected pathwidth can be~$\log^{1+o(1)}(n)$-approximated in polynomial time~\cite{BKS24}.

\paragraph*{Overview.}
In this work, we initiate the study of exponential-time approximation algorithms for vertex-ordering problems.
A recurring theme in our algorithms is a ``balanced-cut'' strategy, which results in a running time that is strictly faster than $O^*(2^n)$ and which we consider our main conceptual contribution.
We briefly describe the main idea here.
We first show how to find a partition of the set of vertices in a graph into two sets~$L$ and~$R$ of given sizes such that the number of arcs from~$R$ to~$L$ is (approximately) minimized, that is, we solve \dkmc.\footnote{We mention that we do not always just consider a single partition. Sometimes, we iterate over many such partitions and return the best solution over all considered partitions.}
Depending on the specific problem at hand, we then either solve the two instances induced by these sets of vertices exactly or approximately.
Finally, we show how to combine the two partial solutions with the computed cut between~$R$ and~$L$.
For \textsc{Feedback Arc Set}, we design an additional ``self-improving'' $(1+\varepsilon)$-approximation algorithm by combining the above approach with a novel bootstrapping strategy.

The rest of this work is structured as follows.
In \cref{sec:prelim}, we introduce notation and concepts used in this work.
We show our results for unweighted problems in \cref{sec:unweighted} and for weighted problems in \cref{sec:weighted}.
We conclude with \cref{sec:concl}.

\section{Preliminaries}
\label{sec:prelim}
For a positive integer~$n \in \N$, we denote by~$[n]$ the set~$\{1,2,\ldots,n\}$.
We use standard graph-theoretic notation.
In particular, for a given graph~$G=(V,E)$ and a set~$V' \subseteq V$ of vertices, we use~$G[V']$ to denote the subgraph of~$G$ induced by~$V'$.
To avoid confusion, we will use the term \emph{edges} for undirected graphs and the term \emph{arcs} for directed graphs.
Moreover, we will refer to the edge set of an undirected input graph by~$E$ and the arc set of a directed input graph by~$A$.
All graphs in this work are simple, that is, they do not contain self-loops or parallel edges/arcs (but we allow in the directed case both the arc~$(u,v)$ and the arc~$(v,u)$ to be present at the same time for a pair of vertices~$u \neq v$).
We always denote the number of vertices in the input graph by~$n$ and the number of edges/arcs by~$m$.
We say that a graph is edge-weighted (respectively arc-weighted) if we are given an additional weight function~$w \colon E \rightarrow \N$ ($w \colon A \rightarrow \N$).
The weight of an edge~$e \in E$ is then~$w(e)$ and the weight of a set~$E' \subseteq E$ of edges is the sum of edge weights of edges in~$E$.

A \emph{ordering} $\pi$ of the vertices of a graph~$G$ with vertex set~$V$ is a bijection between~$V$ and~$|V|$.
For a directed graph~$G=(V,A)$ and an ordering~$\pi$ of~$V$, we say that an arc~$(u,v)$ is forward with respect to~$\pi$ if~$\pi(u) < \pi(v)$ and the arc is backward with respect to~$\pi$ otherwise.
The stretch of an edge~$\{u,v\}$ (or an arc~$(u,v)$) is~$|\pi(u) - \pi(v)|$.
For an integer~$i \in [n]$, we say the cut at position~$i$ with respect to~$\pi$ is~${\cut_{\pi}^i = \{\{u,v\} \in E \mid \pi(u) > i \land \pi(v) \leq i\}}$.
In the directed case, we use the definition~$\cut_{\pi}^i = \{(u,v) \in A \mid \pi(u) > i \land \pi(v) \leq i\}$.
See \cref{fig:cut} for an example.
\begin{figure}
    \centering
    \begin{tikzpicture}
        \node[circle,draw,inner sep=3pt] at(0,0) (1) {};
        \node[circle,draw,inner sep=3pt] at(1,0) (2) {};
        \node[circle,draw,inner sep=3pt] at(2,0) (3) {};
        \node[circle,draw,inner sep=3pt] at(3,0) (4) {};
        \node[circle,draw,inner sep=3pt] at(4,0) (5) {};
        \node[circle,draw,inner sep=3pt] at(5,0) (6) {};

        \draw[->] (1) to (2);
        \draw[->] (2) to (3);
        \draw[->, bend right=30] (1) to (4);
        \draw[->] (4) to (5);
        \draw[->] (5) to (6);
        \draw[->,red,bend right=30] (6) to (3);
        \draw[->,red,bend right=30] (5) to (2);

        \draw[dashed] (.5,1) to (.5,-1);
        \draw[dashed] (2.5,1) to (2.5,-1);
    \end{tikzpicture}
    \caption{An example for~$\cut_{\pi}$ in the directed setting. The vertices are ordered by some ordering~$\pi$ from left to right. The cuts~$\cut_{\pi}^1$ and~$\cut_{\pi}^3$ are depicted by dashed lines. The cut~$\cut_{\pi}^1$ is empty and the cut~$\cut_{\pi}^3$ consists of all the red arcs. The stretch of both red arcs is~3.}
    \label{fig:cut}
\end{figure}
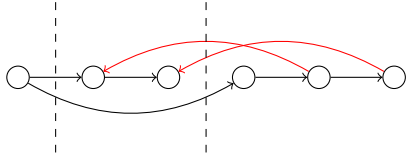

We assume familiarity with the Bachmann–Landau notation (also known as big-O notation) and we use~$O^*$ to hide polynomial factors in the input size.
We note, that this on the one hand hides polynomial factors in~$\log(\wmax)$, where~$\wmax$ is the largest weight in the input graph.
On the other hand, $O^*(\alpha^n) \subseteq O((\alpha+\varepsilon)^{n})$ for any $\alpha \geq 1$ and any~$\varepsilon > 0$.

\section{Unweighted Problems}
\label{sec:unweighted}
In this section, we present our algorithms for unweighted input graphs.
We present all 

algorithms for the directed case but mention that it is straight-forward to adapt all of these algorithms to the undirected case.
Most of our results are based on the following faster algorithm for \dkmc{} which is a generalization of a similar algorithm for \textsc{Maximum Cut} by Williams~\cite{Williams04}.
Here and in the following,~${\omega < 2.373}$ is the matrix-multiplication exponent.

\begin{proposition}
    \label{prop:dkmc}
    \dkmc{} can be solved in~$O^*(\min(2^{\frac{\omega n}{3}},n^{\omega k}))$ time.
\end{proposition}

\begin{proof}
    Let~$(G=(V,A),k)$ be an input instance of \dkmc.
    We partition~$V$ arbitrarily into three sets~$V_1,V_2$, and~$V_3$ of (roughly) equal size.
    Next, we guess in~$O(n^2)$ time how many vertices of~$V_1, V_2$, and~$V_3$ belong to an optimal solution~$L$.
    Let these numbers be~$k_1,k_2$ and~$k_3 = k - k_1 - k_2$.
    Next, we build a graph~$H$ as follows.
    For each~$i \in [3]$, we introduce a vertex~$v_j^i$ for each subset~$V_j^i \subseteq V_i$ with~$|V_j^i| = k_i$.
    Let~${\delta(v_j^i) = |\{(u,w)\in A \mid u \in V_i \setminus V_j^i \land w \in V_j^i\}|}$ be the number of arcs from~$V_i \setminus V_j^i$ to~$V_j^i$.
    Note that~$H$ contains~$O(\min(2^{\nicefrac{n}{3}},\binom{\nicefrac{n}{3}}{k}))$ vertices.
    For each pair of vertices~$v_j^i,v_{j'}^{i'}$ with~${i \neq i'}$, we add an edge between them. We set the weight of this edge to
    \begin{align*}
        |\{(u,w) \in A &\mid u \in V_i \setminus V_j^i \land w \in V_{j'}^{i'}\}|\\
        +\ |\{(u,w) \in A &\mid u \in V_{i'} \setminus V_{j'}^{i'} \land w \in V_{j}^{i}\}| + \frac{\delta(v_j^i)+\delta(v_{j'}^{i'})}{2}.
    \end{align*}
    Note that each edge weight is at most~$n^2$.
    Next, for each pair~$i \neq i'$, we guess the weight of the edge~$\{v_j^i,v_{j'}^{i'}\}$ in a minimum-weight triangle in~$H$ (without guessing the indices~$j$ and~$j'$).
    Note that there are~$n^6$ possible guesses.
    For each guess, we compute in time linear in the size of~$H$ a subgraph which only contains the edges corresponding to the current guess. We then check whether this subgraph contains a triangle in~$O^*((\min(2^{\frac{n}{3}},\binom{\nicefrac{n}{3}}{k}))^\omega) = O^*(\min(2^{\frac{\omega n}{3}},n^{\omega k}))$ time~\cite{IR78}.
    The algorithm also finds a triangle if one exists.
    After going through all possible guesses, we report the triangle of minimum weight.
    Note that this weight corresponds by construction to a minimum directed~$(k,n-k)$-cut such that~${|L \cap V_i|=k_i}$ for all~$i\in [3]$.
    Since we iterate over all possible choices of~$k_i$, we always find a minimum-weight directed~$(k,n-k)$-cut.
    The total running time is in~$O^*(\min(2^{\frac{\omega n}{3}},n^{\omega k}))$.
\end{proof}

\subsection{Feedback Arc Set}

We next show how to use \cref{prop:dkmc} to design a~$(1+\varepsilon)$-approximation for \fas{} for any~$\varepsilon > 0$ that is faster than the~$O^*(2^n)$-time Held-Karp algorithm.
We first show how to find a 2-approximation in~$O^*(2^{\frac{\omega n}{3}})$ time and then show how to use a~$(1+\frac{1}{k})$-approximation to find a~$(1+\frac{1}{k+1})$-approximation.

\begin{lemma}
    \label{cor:2approx}
    \fas{} can be 2-approximated in~$O^*(2^{\frac{\omega n}{3}})$ time.
\end{lemma}

\begin{proof}
    Let~$G=(V,A)$ be an input graph for \fas.
    We first compute an optimal solution~$L \subseteq V$ of the instance~$(G,\lfloor\frac{n}{2}\rfloor)$ of \dkmc{} in~$O^*(2^{\frac{\omega n}{3}})$ time using \cref{prop:dkmc}. Observe that $|L| = \lfloor\frac{n}{2}\rfloor$, and that the number of arcs from~$V\setminus L$ to~$L$ is a lower bound for the size of a minimum feedback arc set in~$G$.
    Let this set of arcs be~$A'$ and let the set of arcs from~$L$ to~$V \setminus L$ be~$A''$.
    Next, we solve the two instances~$G[L]$ and~$G[V\setminus L]$ of \fas{} optimally in~$O^*(2^{\frac{n}{2}})$ time~\cite{HK61}.
    See \cref{fig:fas} for an example.
    Since a feedback arc set in a graph is also a feedback arc set in any subgraph, it holds that the optimal solutions in~$G[L]$ and~$G[V \setminus L]$ are combined a lower bound for the size of a feedback arc set in~$G$.
    Hence, placing the vertices of~$L$ in the computed order before the vertices of~$V \setminus L$ in the computed order yields a 2-approximation as each arc in~$A'$ contributes one and the arcs in~$A''$ do not contribute anything.
    The running time is in~$O^*(2^{\frac{\omega n}{3}} + 2^{\frac{n}{2}}) = O^*(2^{\frac{\omega n}{3}})$.
    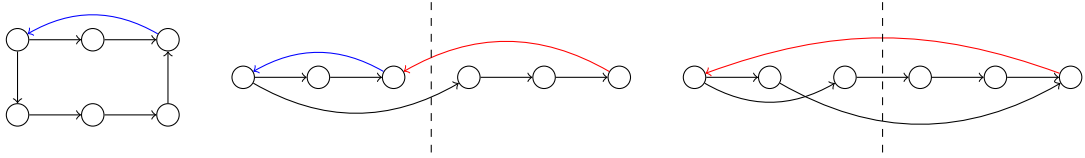
\begin{figure}
        \centering
        \begin{tikzpicture}
            \def\x{3}
            \def\y{9}
            \node[circle,draw,inner sep=3pt] at(0,.5) (o1) {};
            \node[circle,draw,inner sep=3pt] at(1,.5) (o2) {};
            \node[circle,draw,inner sep=3pt] at(2,.5) (o3) {};
            \node[circle,draw,inner sep=3pt] at(0,-.5) (o4) {};
            \node[circle,draw,inner sep=3pt] at(1,-.5) (o5) {};
            \node[circle,draw,inner sep=3pt] at(2,-.5) (o6) {};

            \draw[->] (o1) to (o2);
            \draw[->] (o2) to (o3);
            \draw[->,bend right=30,blue] (o3) to (o1);
            \draw[->] (o1) to (o4);
            \draw[->] (o4) to (o5);
            \draw[->] (o5) to (o6);
            \draw[->] (o6) to (o3);
            
            \node[circle,draw,inner sep=3pt] at(0+\x,0) (1) {};
            \node[circle,draw,inner sep=3pt] at(1+\x,0) (2) {};
            \node[circle,draw,inner sep=3pt] at(2+\x,0) (3) {};
            \node[circle,draw,inner sep=3pt] at(3+\x,0) (4) {};
            \node[circle,draw,inner sep=3pt] at(4+\x,0) (5) {};
            \node[circle,draw,inner sep=3pt] at(5+\x,0) (6) {};

            \draw[->] (1) to (2);
            \draw[->] (2) to (3);
            \draw[->,bend right=30,blue] (3) to (1);
            \draw[->, bend right=30] (1) to (4);
            \draw[->] (4) to (5);
            \draw[->] (5) to (6);
            \draw[->,red,bend right=30] (6) to (3);

            \draw[dashed] (2.5+\x,1) to (2.5+\x,-1);

            \node[circle,draw,inner sep=3pt] at(0+\y,0) (c1) {};
            \node[circle,draw,inner sep=3pt] at(1+\y,0) (c2) {};
            \node[circle,draw,inner sep=3pt] at(2+\y,0) (c3) {};
            \node[circle,draw,inner sep=3pt] at(3+\y,0) (c4) {};
            \node[circle,draw,inner sep=3pt] at(4+\y,0) (c5) {};
            \node[circle,draw,inner sep=3pt] at(5+\y,0) (c6) {};

            \draw[->] (c1) to (c2);
            \draw[->, bend right=30] (c2) to (c6);
            \draw[->,bend right=30] (c1) to (c3);
            \draw[->] (c3) to (c4);
            \draw[->] (c4) to (c5);
            \draw[->] (c5) to (c6);
            \draw[->,red,bend right=20] (c6) to (c1);

            \draw[dashed] (2.5+\y,1) to (2.5+\y,-1);
        \end{tikzpicture}
        \caption{An example of \fas. The input graph is given on the left and the middle/right show two optimal solutions for \dkmc{} with~$k=3$ (red arcs). The optimal solution to \fas{} is shown in blue (and coincides with the red arc on the right). If we use the cut in the middle picture, then we can only find a 2-approximation while the cut on the right leads to an optimal solution.}
        \label{fig:fas}
    \end{figure}
\end{proof}

We next show how to use any constant-factor approximation for \fas{} that is faster than~$O^*(2^n)$ to compute a better constant factor approximation that is also faster than~$O^*(2^n)$.
This ``boosting technique'' results in the following approximation scheme.

\begin{theorem}
    \label{thm:fas}
    For each constant~$\varepsilon > 0$, there is a~$\delta_\varepsilon > 0$ such that \fas{} can be~$(1+\varepsilon)$-approximated in~$O((2-\delta_\varepsilon)^n)$ time.
\end{theorem}

\begin{proof}
    We prove this statement by showing via induction that for each~$k \geq 1$, there is a~$\delta_k > 0$ such that \fas{} can be~$(1+\frac{1}{k})$-approximated in~$O((2-\delta_k)^n)$ time.
    Since any constant~$\varepsilon > 0$ is larger than~$\frac{1}{k}$ for some constant value of~$k$, this will prove the theorem.
    Note that the base case~$k=1$ is proven by \cref{cor:2approx}.
    It remains to show the induction step, that is, assuming that the statement is true for some value of~$k$, we need to show that the statement is also true for~$k+1$.
    To this end, let~$k$ be arbitrary but fixed and assume that \fas{} can be~$(1+\frac{1}{k})$-approximated in~$(2-\delta_k)^n$ time for some~$\delta_k > 0$.
    Let~$\gamma$ be the solution to the equation \[\frac{\gamma \log(\gamma) - (\gamma - 1) \log(\gamma - 1)}{\gamma - 1} = 1 - \log(2-\delta_k),\] where all logarithms use base 2, and let~$\alpha = \frac{1}{\gamma}$.\footnote{We note that~$\gamma$ is roughly~$- \frac{W(-\frac{(2-\delta_k)(1-\log(2-\delta_k))}{2})}{1-\log(2-\delta_k)}$, where~$W$ is the Lambert~$W$ function (also known as the product logarithm). It is known that this function cannot be expressed in terms of elementary functions.}
    Note that~$\gamma \geq 2$ exists for each~$0 < \delta_k < 1$ as the right side of the equations is some value between zero and one and the left side evaluates to~$2$ for~$\gamma = 2$ and tends towards~$0$ for increasing values of~$\gamma$ as shown next.
    Note that the left side of the previous equation is equivalent to~$(\log(\gamma) - \log(\gamma - 1)) + \frac{\log(\gamma)}{\gamma - 1}$.
    The first summand describes the derivative of the logarithm (which is~$\frac{1}{\ln(2)\gamma}$) and therefore tends towards 0 for increasing values of~$\gamma$.
    The second summand also clearly tends towards zero as the logarithm grows sublinearly.

    We next describe the~$(1+\frac{1}{k+1})$-approximation for \fas.
    Let~$G = (V,A)$ be the input graph and let~$n$ be the number of vertices in~$V$.
    For the sake of simplicity, we will assume that~$\alpha n$ is an integer to avoid dealing with rounding issues (which will vanish in the Bachmann-Landau notation anyway).
    Bodlaender et al.~\cite{BFK+12} showed how to solve \fas{} exactly in~$O^*(2^n)$ time.
    Importantly, their algorithm iterates over all induced subgraphs (in order of increasing size), and computes the optimal solution for each induced subgraph in polynomial time for each subgraph.
    We use their algorithm for all induced subgraphs with at most~$\alpha n$ vertices.
    Let~$\mathcal{V}$ be the set of all sets of vertices of size exactly~$\alpha n$.
    Using the algorithm by Bodlaender et al.~\cite{BFK+12} to compute the optimal solution for all graphs induced by sets in~$\mathcal{V}$ takes~$O^*(\binom{n}{\alpha n})$ time.
    In the same time, we can also compute for each~$V' \in \mathcal{V}$ the number of arcs in~${\{(u,v) \mid u \in V \setminus V' \land v \in V'\}}$.
    Let~$a(V')$ denote this number and let~${V^* \in \mathcal{V}}$ be a set minimizing~$a(V^*)$.
    We then use the algorithm by Bodlaender et al.~\cite{BFK+12} to solve the instance~${G[V \setminus V^*]}$ optimally in~$O^*(2^{(1-\alpha)n})$ time and for all other~$V' \in \mathcal{V}$, we use the~$(1 + \frac{1}{k})$-approximation for~$G[V \setminus V']$ that runs in~$O((2-\delta_k)^{(1-\alpha)n})$ time.
    Finally for each~$V' \in \mathcal{V}$, we compute the sum of the optimal value for~$G[V']$ plus~$a(V')$ plus the size of the computed solution for~$G[V \setminus V']$.
    We return the smallest number found (or the set associated to the smallest number if we want to return the feedback arc set).

    It remains to show that our algorithm always produces a~$(1+\frac{1}{k+1})$-approximation and to analyze the running time.
    For the approximation factor, let~$F$ be a feedback arc set of~$G$ of minimum size and let~$\opt = |F|$.
    We distinguish between the following two cases.
    Either~${a(V^*) \leq \frac{1}{k+1} \opt}$ or~$a(V^*) > \frac{1}{k+1} \opt$.
    In the first case, note that the solution we computed for~$V^*$ consists of optimal solutions for~$G[V^*]$ and~$G[V\setminus V^*]$ plus~$a(V^*)$.
    Since the first two summands are combined at most~$\opt$ and~$a(V^*) \leq \frac{1}{k+1} \opt$, we indeed return a feedback arc set of size at most~$(1 + \frac{1}{k+1})\opt$.
    If~$a(V^*) > \frac{1}{k+1} \opt$, then note that for each~$V' \in \mathcal{V}$ it holds that~$a(V') > \frac{1}{k+1} \opt$.
    Moreover, since the graph~$G' = (V,A\setminus F)$ is acyclic, it has a topological ordering.
    Let~$V''$ be the set of the first~$\alpha n$ vertices in that topological ordering.
    Let~$d$ be the size of a minimum feedback arc set in~$G[V \setminus V'']$ and let~$c = \frac{d}{\opt}$.
    Note that~$c \leq \frac{k}{k+1}$.
    Finally, the solution we compute for~$V''$ has size \[(1-c)\opt + (1+\frac{1}{k}) c \opt \leq \opt - c\opt + c\opt + \frac{1}{k} \frac{k}{k+1} \opt = (1+\frac{1}{k+1})\opt.\]
    Thus, we also return a~$(1+\frac{1}{k+1})$-approximation in this case.

    Let us now analyze the running time.
    Recall that~$\gamma = \frac{1}{\alpha}$.
    Computing~$\mathcal{V}$,~$a(V')$, and the optimal solution for~$G[V']$ for each~$V' \in \mathcal{V}$ take~$O^*(\binom{n}{\alpha n})$ time.
    Solving the instance~${G[V \setminus V^*]}$ optimally takes~$O^*(2^{(1-\alpha)n})$ time.
    Computing the approximate solutions for all sets in~$\mathcal{V}$ takes overall~${O^*(\binom{n}{\alpha n} \cdot (2-\delta_k)^{(1-\alpha)n})}$~time.
    By standard arguments,~$\binom{n}{\alpha n} \in O((\frac{\gamma^\gamma}{(\gamma-1)^{\gamma-1}})^{\alpha n})$~\cite{Rob55}.
    Hence
    \begin{align*}
        O(\binom{n}{\alpha n} \cdot (2-\delta_k)^{(1-\alpha)n}) &= O(2^{\log(\frac{\gamma^\gamma}{(\gamma-1)^{\gamma-1}})^{\alpha n})} \cdot 2^{\log((2-\delta_k)^{(1-\alpha)n}})\\
        &\subseteq O(2^{\alpha n (\log(\gamma^\gamma) - \log((\gamma-1)^{\gamma-1})) + (1-\alpha)n\log(2-\delta_k)})\\
        &\subseteq O(2^{(1-\alpha) n \big( \frac{\alpha}{1-\alpha} (\gamma \log(\gamma) - (\gamma - 1) \log(\gamma-1)) + \log(2-\delta_k)\big)})\\
        &\subseteq O(2^{(1-\alpha) n \big( \frac{\gamma \log(\gamma) - (\gamma - 1) \log(\gamma-1)}{\gamma - 1}  + \log(2-\delta_k)\big)})\\
        &\subseteq O(2^{(1-\alpha)n}),
    \end{align*}
    where again all logarithms use base~$2$.
    The last equality is due to the definition of~$\gamma$.
    Thus, the total running time of our algorithm is~$O^*(2^{(1-\alpha)n})$.
    This is in~$O((2-\delta_{k+1})^n)$ time for any~${0 < \delta_{k+1} < 2-2^{1-\alpha}}$.
    This completes the induction step and concludes the proof.
\end{proof}

\subsection{Cutwidth}

We continue with \textsc{Directed Cutwidth}.
We again show how to use \cref{prop:dkmc} to compute a 2-approximaiton.

\begin{proposition}
    \label{prop:cw}
    \textsc{Directed Cutwidth} can be 2-approximated in~$O^*(2^{\frac{\omega n}{3}})$ time.
\end{proposition}

\begin{proof}
    Let~$G=(V,A)$ be an input to \textsc{Directed Cutwidth}.
    We first find an optimal solution~$L$ to the instance~$(G,\lfloor \frac{n}{2} \rfloor)$ of \dkmc{} in~$O^*(2^{\frac{\omega n}{3}})$ time using \cref{prop:dkmc}.
    Let~$R = V \setminus L$ and let~$A^*$ be the set of arcs~$(u,v)$ with~$u \in R$ and~$v \in L$.
    Note that the size of~$A^*$ is by definition a lower bound for the directed cutwidth of~$G$.
    Next, we can solve~$G[L]$ and~$G[R]$ optimally in~$O^*(2^{\frac{n}{2}})$ time~\cite{HK61}.
    Since adding a vertex and/or arc to a graph cannot decrease its directed cutwidth, we have that the cutwidth of~$G[L]$ and~$G[R]$ are both at most the cutwidth of~$G$.
    Hence, placing all vertices in~$L$ according to the computed solution for~$G[L]$ and then all vertices in~$R$ according to the solution for~$G[R]$ gives an ordering for the graph~$G'=(V,A \setminus A^*)$ that shows that the cutwidth of~$G'$ is at most the cutwidth of~$G$.
    See \cref{fig:cutwidth} for an example.
    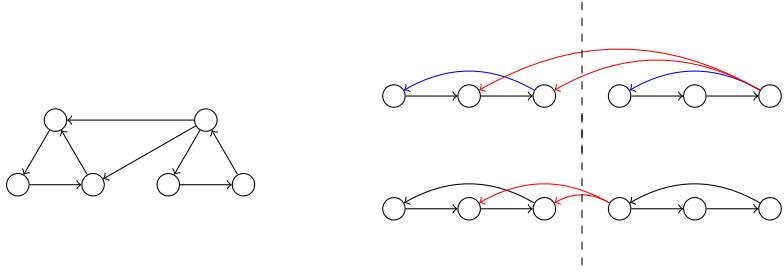
\begin{figure}
        \centering
        \begin{tikzpicture}
            \def\x{5}
            \def\y{5}
            \def\z{.75}
            \node[circle,draw,inner sep=3pt] at(0,-.43) (o1) {};
            \node[circle,draw,inner sep=3pt] at(1,-.43) (o2) {};
            \node[circle,draw,inner sep=3pt] at(.5,.43) (o3) {};
            \node[circle,draw,inner sep=3pt] at(2,-.43) (o4) {};
            \node[circle,draw,inner sep=3pt] at(3,-.43) (o5) {};
            \node[circle,draw,inner sep=3pt] at(2.5,.43) (o6) {};

            \draw[->] (o1) to (o2);
            \draw[->] (o2) to (o3);
            \draw[->] (o3) to (o1);
            \draw[->] (o4) to (o5);
            \draw[->] (o5) to (o6);
            \draw[->] (o6) to (o4);
            \draw[->] (o6) to (o2);
            \draw[->] (o6) to (o3);

            \node[circle,draw,inner sep=3pt] at(0+\x,0+\z) (p1) {};
            \node[circle,draw,inner sep=3pt] at(1+\x,0+\z) (p2) {};
            \node[circle,draw,inner sep=3pt] at(2+\x,0+\z) (p3) {};
            \node[circle,draw,inner sep=3pt] at(3+\x,0+\z) (p4) {};
            \node[circle,draw,inner sep=3pt] at(4+\x,0+\z) (p5) {};
            \node[circle,draw,inner sep=3pt] at(5+\x,0+\z) (p6) {};

            \draw[->] (p1) to (p2);
            \draw[->] (p2) to (p3);
            \draw[->,bend right=30,blue] (p3) to (p1);
            \draw[->] (p4) to (p5);
            \draw[->] (p5) to (p6);
            \draw[->,bend right=30,blue] (p6) to (p4);
            \draw[->,bend right=30,red] (p6) to (p2);
            \draw[->,bend right=30,red] (p6) to (p3);

            \draw[dashed] (2.5+\x,1.25+\z) to (2.5+\x,-.75+\z);

            \node[circle,draw,inner sep=3pt] at(0+\y,0-\z) (q1) {};
            \node[circle,draw,inner sep=3pt] at(1+\y,0-\z) (q2) {};
            \node[circle,draw,inner sep=3pt] at(2+\y,0-\z) (q3) {};
            \node[circle,draw,inner sep=3pt] at(3+\y,0-\z) (q4) {};
            \node[circle,draw,inner sep=3pt] at(4+\y,0-\z) (q5) {};
            \node[circle,draw,inner sep=3pt] at(5+\y,0-\z) (q6) {};

            \draw[->] (q1) to (q2);
            \draw[->] (q2) to (q3);
            \draw[->,bend right=30] (q3) to (q1);
            \draw[->] (q4) to (q5);
            \draw[->] (q5) to (q6);
            \draw[->,bend right=30] (q6) to (q4);
            \draw[->,bend right=30,red] (q4) to (q2);
            \draw[->,bend right=30,red] (q4) to (q3);
            
            \draw[dashed] (2.5+\y,1.25-\z) to (2.5+\y,-.75-\z);
        \end{tikzpicture}
        \caption{An example of \textsc{Directed Cutwidth}. The input instance is depicted on the left with the exception that all vertices in the left triangle have arcs towards all vertices in the right triangle unless the respective reverse arc is present. We decided to not show these arcs for the sake of readibility. On the right, two optimal solutions to \dkmc{} with~$k=3$ are shown (red arcs) with optimal orderings for the two respective subproblems of \textsc{Directed Cutwidth} are depicted. The ordering shown on the top has a directed cutwidth of~3 (e.g. at positions 4 and 5) and the ordering on the bottom has a directed cutwidth of~$2$.}
        \label{fig:cutwidth}
    \end{figure}
    Adding the set~$A^*$ to~$G'$ can increase the cutwidth by at most~$|A^*|$ and adding arcs~$(u,v)$ with~$u \in L$ and~$v \in R$ does not increase the directed cutwidth.
    Thus, the computed ordering is a 2-approximation as the size of~$A^*$ is at most the cutwidth of~$G$ as shown above.
    Note that the running time is in~$O^*(2^{\frac{\omega n}{3}} + 2^{\frac{n}{2}}) = O^*(2^{\frac{\omega n}{3}})$.
    This concludes the proof.
\end{proof}

Unfortunately, due to the inherent maxima nature of \textsc{Directed Cutwidth}, we do not know how to apply a boosting technique similar to \cref{thm:fas} to get an approximation factor smaller than~$2$.

\subsection{Optimal Linear Arrangement}

In this section, we study \ola.
We show how to compute a~$(2+\varepsilon)$-approximation for \dola{} for each~$\varepsilon > 0$ in~${O((2-\delta_\varepsilon)^n)}$ time for some~$\delta_\varepsilon > 0$  only depending on~$\varepsilon$.
We then show how to improve the algorithm in the undirected setting to achieve a~$(\frac{3}{2}+\varepsilon)$-approximation in the same time.

\begin{proposition}
    \label{prop:dola}
    For every~$0 < \alpha < 1$, there is a~$(1+\frac{1}{(1-\alpha)})$-approximation for \dola{} that runs in~$O^*(2^{\frac{\omega n}{3}} + 2^{(1-\frac{\alpha}{2})n})$ time, where~$\omega$ is the matrix-multiplication exponent.
\end{proposition}

\begin{proof}
    Let~$G=(V,A)$ be an input to \dola.
    We use \cref{prop:dkmc} to solve \dkmc{} for each~${\frac{\alpha n}{2} \leq k \leq n - \frac{\alpha n}{2}}$.
    This takes~$O^*(2^{\frac{\omega n}{3}})$ time.    
    Let~$k'$ be a value where the number of arcs from~$V \setminus L$ to~$L$ are minimum for some set~$L$ of size~$k'$.
    Let~$L'$ be a set of vertices of size~$k'$ minimizing the number of arcs from~${R' = V \setminus L'}$ to~$L'$ and let~$A'$ be the set of arcs from~$R'$ to~$L'$.
    By the pigeonhole principle, the size of~$A'$ is at most~$\frac{\opt}{(1-\alpha)n}$ where~$\opt$ is the optimal directed-optimal-linear-arrangement value for~$G$ (the value we want to minimize in \dola).
    Note that~$|L'|,|R'| \leq (1 - \frac{\alpha}{2})n$.
    Next, we solve the instances~$G[L]$ and~$G[R]$ optimally in~$O^*(2^{(1 - \frac{\alpha}{2})n})$~time using the algorithm by Bodlaender et al.~\cite{BFK+12}.
    To get an approximation for the whole graph~$G$, we first order all vertices in~$L$ according to the computed solution for~$G[L]$ and then all vertices in~$R$ according to the computed solution for~$R$.
    Note that this is an optimal solution for the graph~$G'=(V,A\setminus A')$.
    It remains to add the arcs between~$L$ and~$R$.
    Note that each edge in~$A'$ can contribute at most~$n$ and therefore the edges in~$A'$ can contribute at most~$\frac{\opt}{1-\alpha}$ in total.
    Arcs from~$L$ to~$R$ do not contribute as they are forward arcs in the computed ordering.
    Added to the cost for~$G'$ (which is at most~$\opt$), this yields the claimed approximation factor.
    The running time is in~$O^*(2^{\frac{\omega n}{3}} + 2^{(1 - \frac{\alpha}{2})n})$.
\end{proof}

We mention that the above result yields a~$2.72$-approximation in~$O^*(2^{\frac{\omega n}{3}})$~time, assuming the currently known value of $\omega < 2.372$.
We next show how to improve the approximation factor to~$(\frac{3}{2}+\varepsilon)$ in the undirected setting.
The general strategy is the same but in the end, we use the fact that an ordering and its reverse have the same value in terms of \ola{} and leverage this fact to argue that  contribution of each edge between the two parts is at most~$\frac{n}{2}$ on average.

\begin{proposition}
    For each constant~$0 < \alpha < 1$, there is a~$(1+\frac{1}{2(1-\alpha)})$-approximation for \ola{} that runs in~$O^*(2^{\frac{\omega n}{3}} + 2^{(1-\frac{\alpha}{2})n})$ time, where~$\omega$ is the matrix-multiplication exponent.
\end{proposition}

\begin{proof}
    Let~$G=(V,E)$ be an input graph for \ola.
    We start the same as in \cref{prop:dola}, that is, we solve \textsc{Minimum $(k,n-k)$-Cut} for each~${\frac{\alpha n}{2} \leq k \leq \frac{n}{2}}$, define sets~$L', R'$, and~$E'$ and value~$\opt$, and solve~$G[L']$ and~$G[R']$ optimally in~${O^*(2^{\frac{\omega n}{3}} + 2^{(1 - \frac{\alpha}{2})n})}$ time overall.
    Let~$i = |L'|$.
    We next test for both~$L'$ and~$R'$ whether the average length of edges in~$E'$ is smaller if we order the vertices according to the computed solution or in the reverse order.
    To this end, we pretend that the endpoint outside the respective part is fixed at position~$i$.
    See \cref{fig:ola} for an example.
    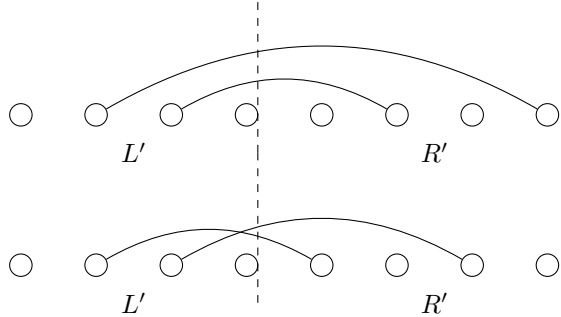
\begin{figure}
        \centering
        \begin{tikzpicture}
            \def\x{0}
            \def\y{-2}
            \node[circle,draw,inner sep=3pt] at(0,0) (1) {};
            \node[circle,draw,inner sep=3pt] at(1,0) (2) {};
            \node[circle,draw,inner sep=3pt] at(2,0) (3) {};
            \node[circle,draw,inner sep=3pt] at(3,0) (4) {};
            \node[circle,draw,inner sep=3pt] at(4,0) (5) {};
            \node[circle,draw,inner sep=3pt] at(5,0) (6) {};
            \node[circle,draw,inner sep=3pt] at(6,0) (7) {};
            \node[circle,draw,inner sep=3pt] at(7,0) (8) {};

            \node at(1.5,-.5) {$L'$};
            \node at(5.5,-.5) {$R'$};

            \draw[bend right=30] (8) to (2);
            \draw[bend right=30] (6) to (3);

            \draw[dashed] (3.15,1.5) to (3.15,-.5);

            \node[circle,draw,inner sep=3pt] at(0+\x,0+\y) (1) {};
            \node[circle,draw,inner sep=3pt] at(1+\x,0+\y) (2) {};
            \node[circle,draw,inner sep=3pt] at(2+\x,0+\y) (3) {};
            \node[circle,draw,inner sep=3pt] at(3+\x,0+\y) (4) {};
            \node[circle,draw,inner sep=3pt] at(4+\x,0+\y) (5) {};
            \node[circle,draw,inner sep=3pt] at(5+\x,0+\y) (6) {};
            \node[circle,draw,inner sep=3pt] at(6+\x,0+\y) (7) {};
            \node[circle,draw,inner sep=3pt] at(7+\x,0+\y) (8) {};

            \node at(1.5+\x,-.5+\y) {$L'$};
            \node at(5.5+\x,-.5+\y) {$R'$};

            \draw[bend right=30] (5) to (2);
            \draw[bend right=30] (7) to (3);

            \draw[dashed] (3.15+\x,1.5+\y) to (3.15+\x,-.5+\y);
        \end{tikzpicture}
        \caption{An example of \ola{} where the vertices are divided into two parts~$L'$ and~$R'$ with two edges between the two parts and some ordering within the two parts is fixed. For the sake of simplicity, we do not show edges within the two parts.
        The two orderings only differ in that we reverse the ordering of the vertices in~$R'$.
        The average length within~$L'$ of the two edges between~$L'$ and~$R'$ is 1.5 in both orderings. The average length within~$R'$ is 3 in the above ordering and~$2'$ in the below ordering.}
        \label{fig:ola}
    \end{figure}
    We pick for each of the two the ordering which results in shorter average lengths and observe that the average length of an edge in~$E'$ is now at most~$\frac{n}{2}$ as the average length ``in~$L'$'' is at most~$\frac{|L'|-1}{2}$, the average length ``in~$R'$'' is at most~$\frac{|R'|+1}{2}$, and~$|L'| + |R'| = n$.
    The combined cost for all edges in~$E'$ is therefore at most~$\frac{\opt}{(1-\alpha)n} \cdot \frac{n}{2} = \frac{\opt}{2(1 - \alpha)}$ and added to the cost for~$G'$ (which is at most~$\opt$), this yields the claimed approximation factor.
    The running time is in~$O^*(2^{\frac{\omega n}{3}} + 2^{(1 - \frac{\alpha}{2})n})$.
\end{proof}

The above results yields a~$1.86$-approximation in~$O^*(2^{\frac{\omega n}{3}})$~time using the currently known value of $\omega \approx 2.372$.

\subsection{Pathwidth}

We conclude this section with a $2$-approximation for \textsc{Directed Pathwidth} in~$O(1.66^n)$ time.
We point out that this result stands out from the previous in two ways: \textsc{Directed Pathwidth} can be solved in~$O(1.89^n)$ time and our algorithm is not based on \cref{prop:dkmc}.

\begin{theorem}
    \label{thm:pw}
    \textsc{Directed Pathwidth} can be 2-approximated in~$O(1.66^n)$ time.
\end{theorem}

\begin{proof}
    Let~$G=(V,A)$ be the input graph for which we want to approximate the pathwidth, let~${\alpha \approx 0.204}$ be the solution to the equation~$1.89^{1-\alpha} = (\frac{1}{\alpha})^\alpha(\frac{1}{1-\alpha})^{1-\alpha}$, and let~$\opt$ be the (unknown) directed pathwidth of~$G$.
    We first use the algorithm by Bodlaender et al.~\cite{BFK+12} to find a set~$V'$ of~$\alpha n$ vertices (and an ordering of these vertices) such that independent of the ordering of the remaining~$(1-\alpha)n$ vertices, it holds for each position~$i$ that at most~$\opt$ vertices before position~$\min(i,\alpha n)$ have in-neighbors after position~$i$.
    This takes~${O^*((\frac{1}{\alpha})^{\alpha n}(\frac{1}{1-\alpha})^{(1-\alpha)n})}$ time (which is~$O^*(1.89^{(1-\alpha)n})$ by the definition of~$\alpha$).
    Next, we solve the instance~$G[V \setminus V']$ of \textsc{Directed Pathwidth} optimally in~$O^*(1.89^{(1-\alpha)n})$ time~\cite{KKKTT16}.
    We return the ordering that places the vertices in~$V'$ first (in the order that was computed before) and then all vertices in~$V \setminus V'$ in the order corresponding to the solution for~$G[V \setminus V']$.
    Note that it holds for each position~$i > \alpha n$ that at most~$\opt$ vertices before position~$\alpha n$ have in-neighbors after position~$i$ and at most~$\opt$ vertices between position~$\alpha n +1$ and~$i$ have in-neighbors after position~$i$, that is, the computed ordering is a $2$-approximation.
    The running time is in~${O^*(1.89^{(1-\alpha)n}) \approx O^*(1.89^{(0.796)n}) \subseteq O(1.66^n)}$.
    This concludes the proof.
\end{proof}

\section{Weighted Problems}
\label{sec:weighted}
In this section, we consider weighted generalizations of the problems studied in the last section.
We denote the largest weight in the input by~$\wmax$.
We mention that if~$\wmax$ is polynomially upper-bounded in~$n$, then all of our previous results generalize in a straight-forward way.
Hence, we mainly focus on cases where~$\wmax$ is large compared to~$n$.
An assumption that is often made when dealing with weighted problems is that all numbers can be added and subtracted in constant time.
However, in the realm of very large weights, this assumption is highly unrealistic as even reading the weight from input takes~$O(\log(\wmax))$ time (which we sometimes hide in the~$O^*$-notation).
We decided to not make the above assumption which basically adds an additional~$O(\log(\wmax))$ factor to all running times.

We mention that we do not study a weighted version of \textsc{Directed Pathwidth}.
To the best of our knowledge, this problem has only been studied once~ \cite{MT09} and our algorithm behind \cref{thm:pw} computes a 2-approximation even in the weighted setting if we have an exact (slower) algorithm.
We suspect that the~$O(1.89^n)$-time algorithm by Kitsunai et al.~\cite{KKKTT16} generalizes to the weighted setting but this was never proven and proving it here would go beyond the scope of this paper.

We start with an arc-weighted generalization of \dkmc.
We mention that if this generalization can be solved in~$O^*(2^{\frac{\omega n}{3}})$ time, then our previous algorithms for \fas, \dola, and \textsc{Directed Cutwidth} can all easily be generalized to the weighted setting.
Unfortunately, even for (undirected) \textsc{Weighted Max Cut}, no such algorithm is known.
However, Williams~\cite{Williams04} showed how to compute a~$(1+\varepsilon)$-approximation for \textsc{Weighted Max Cut} in~$O^*(2^{\frac{\omega n}{3}} (\frac{1}{\varepsilon})^3)$ time.
We will generalize this result to \wdkmc{} and show how to use this approximation for the other problems.

\begin{proposition}
    \label{prop:wdkmc}
     For any~$\varepsilon>0$, \wdkmc{} can be~$(1+\varepsilon)$-approximated in~${O(2^{\frac{\omega n}{3}} \cdot \log(\wmax) \cdot (\frac{1}{\varepsilon})^2\cdot k^2)}$ time.
\end{proposition}

\begin{proof}
    Let~$(G=(V,A),k)$ be an input instance and let~$\varepsilon>0$ be an arbitrary but fixed value.
    We partition~$V$ arbitrarily into three sets~$V_1,V_2$, and~$V_3$ of (roughly) equal size.
    Next, we guess in~$O(k^2)$ time how many vertices of~$V_1, V_2$, and~$V_3$ belong to an optimal solution~$A$.
    Let these numbers be~$k_1,k_2$ and~$k_3 = k - k_1 - k_2$.
    Next, we build an undirected graph~$H$ with edge weights as follows.
    For each~$i \in [3]$, we introduce a vertex~$v_j^i$ for each subset~$V_j^i \subseteq V_i$ with~$|V_j^i| = k_i$.
    Note that the number of vertices in~$H$ is in~$O(2^{\nicefrac{n}{3}})$.
    For each vertex~$v_j^i$, let~$d(v_j^i)$ denote the sum of weights of all arcs from vertices in~$V^i_j$ to vertices in~$V_i \setminus V_j^i$ in~$G$.
    For each pair of vertices~$v_j^i,v_{j'}^{i'}$ with~$i \neq i'$, we add an edge between them of weight
    \[\frac{d(v_j^i)+d(v_{j'}^{i'})}{2} + \sum_{\{(u,w) \in A \mid u \in V_j^i \land w \in V_{i'} \setminus V_{j'}^{i'}\} \cup \{(u,w) \in A \mid u \in V_{j'}^{i'} \land w \in V_{i} \setminus V_{j}^{i}\}}w((u,w)).\]
    Using an algorithm by Zwick~\cite{Zwick02} (see also Williams~\cite{Williams04}), we find a~$(1+\varepsilon)$-approximation for the minimum-weight triangle in~$H$ in~${O(2^{\frac{\omega n}{3}} \frac{1}{\varepsilon} \log(\frac{n^2\wmax}{\varepsilon})) \subseteq O(2^{\frac{\omega n}{3}} (\frac{1}{\varepsilon})^2 \log(\wmax))}$ time.
    Note that~$\log(n^2\wmax) \leq \log(\wmax^3) \leq 3\log(\wmax)$ as we assume~$\wmax > n$.
    Moreover, each triangle in~$H$ corresponds to a directed~$(k,n-k)$-cut in~$H$ of the same weight and each directed~$(k,n-k)$-cut with~$k_i$ vertices in~$V_i$ for each~$i \in [3]$ corresponds to a triangle of the same weight in~$H$. 
    Hence, we compute a~$(1+\varepsilon)$-approximation for \wdkmc.
    The total running time is in~${O(2^{\frac{\omega n}{3}} \log(\wmax) (\frac{1}{\varepsilon})^2 k^2)}$ and this concludes the proof.
\end{proof}

\subsection{Feedback Arc Set}

We next show how to use \cref{prop:wdkmc} to get faster approximation algorithms for the weighted variants of the other vertex-ordering problems we consider.
We start with \wfas.
Similar to the unweighted case, we first show a constant factor approximation and then show a boosting technique to achieve any~$(1+\varepsilon)$-approximation.

\begin{proposition}
    \label{cor:3approx}
    \wfas{} admits a 3-approximation that can be computed in~$O(2^{\frac{\omega n}{3}} \log(\wmax) k^2)$ time.
\end{proposition}
    
\begin{proof}
    Let~$G=(V,A)$ be an input graph for \wfas.
    We first compute a 2-approximation~$L$ for the instance~$(G,\lfloor\frac{n}{2}\rfloor)$ of \wdkmc{} in~$O(2^{\frac{\omega n}{3}} \log(\wmax) k^2)$~time using \cref{prop:wdkmc} with~$\varepsilon = 1$.
    Note that half the sum of weights of arcs from~$R = V\setminus L$ to~$L$ is a lower bound for the size of a minimum feedback arc set.
    We refer to \cref{fig:fas} for an illustration and mention that this time the weight of all red arcs is a $2$-approximation for an optimal cut.
    Next, we solve the two instances~$G[L]$ and~$G[R]$ of \wfas{} optimally in~${O(2^{\frac{n}{2}} n^2 \log(\wmax))}$~time~\cite{HK61}.
    Placing the vertices of~$L$ in the computed order before the vertices of~$R$ in the computed order yields a~\mbox{3-approximation} as the weights of the feedback arcs within both instances are combined also a lower bound for the weight of a minimum-weight feedback arc set in~$G$.
    The running time is in~${O(2^{\frac{\omega n}{3}} \log(\wmax) k^2)}$ as~$\frac{\omega}{3} > \frac{1}{2}$.
\end{proof}

We next present the boosting technique in the weighted setting.
It is very similar to the unweighted setting with two minor adjustments: the formulas change slightly due to the fact that we start with a 3-approximation rather than a 2-approximation and an additional factor of~$O(\log(\wmax))$ is used when using the algorithm by Bodlaender et al.~\cite{BFK+12}.
We present the entire proof for the sake of completeness.

\begin{theorem}
    For each~$\varepsilon > 0$, there is a~$\delta > 0$ such that \wfas{} can be~$(1+\varepsilon)$-approximated in~$O((2-\delta)^n \cdot \log(\wmax))$ time.
\end{theorem}

\begin{proof}
    We prove this statement by showing via induction that for each~$k \geq 1$, there is a~$\delta > 0$ such that \wfas{} can be~$(1+\frac{2}{k})$-approximated in~$O((2-\delta)^n\log(\wmax))$ time.
    Since any~$\varepsilon > 0$ is larger than~$\frac{2}{k}$ for some value of~$k$, this will prove the theorem.
    Note that the base case~$k=1$ is proven by \cref{cor:3approx}.
    It remains to show the induction step, that is, assuming that the statement is true for some value of~$k$, we need to show that the statement is also true for~$k+1$.
    To this end, let~$k$ be arbitrary but fixed and assume that \wfas{} can be~$(1+\frac{2}{k})$-approximated in~$(2-\delta_k)^n$ time for some~$\delta_k > 0$.
    Let~$\gamma$ be the solution to the equation \[\frac{\gamma \log(\gamma) - (\gamma - 1) \log(\gamma - 1)}{\gamma - 1} = 1 - \log(2-\delta_k),\] where all logarithms use base 2, and let~$\alpha = \frac{1}{\gamma}$.
    
    We next describe how to get the~$(1+\frac{2}{k+1})$-approximation for \wfas.
    Let~${G = (V,A)}$ be the input graph and let~$n$ be the number of vertices in~$V$.
    For the sake of simplicity, we will again assume that~$\alpha n$ is an integer.
    Let~$\mathcal{V}$ be the set of all vertices with exactly~$\alpha n$ vertices.
    Using the algorithm by Bodlaender et al.~\cite{BFK+12}, we compute the optimal solution for all graphs induced by sets in~$\mathcal{V}$.
    This takes~$O(\binom{n}{\alpha n} \poly(n) \log(\wmax))$ time.
    In the same time, we can also compute for each~$V' \in \mathcal{V}$ the sum of weights of arcs in~${\{(u,v) \mid u \in V \setminus V' \land v \in V'\}}$.
    Let~$a(V')$ denote this number and let~$V^* \in \mathcal{V}$ be a set minimizing~$a(V^*)$.
    We then use the algorithm by Bodlaender et al.~\cite{BFK+12} to solve the instance~$G[V \setminus V^*]$ optimally in~$O(2^{(1-\alpha)n} \poly(n) \log(\wmax))$ time and for all other~$V' \in \mathcal{V}$, we use the~$(1 + \frac{2}{k})$-approximation for~$G[V \setminus V']$ that runs in~${O((2-\delta_k)^{(1-\alpha)n})}$ time.
    Finally, we for each~$V' \in \mathcal{V}$, we compute the sum of the optimal value for~$G[V']$ plus~$a(V')$ plus the size of the computed solution for~$G[V \setminus V']$.
    We return the smallest number found.
    
    It remains to show that our algorithm always produces a~$(1+\frac{2}{k+1})$-approximation and to analyze the running time.
    For the approximation factor, let~$F$ be a feedback arc set of~$G$ of minimum weight~$\opt$.
    We distinguish between the following two cases.
    Either~${a(V^*) \leq \frac{2}{k+1} \opt}$ or~$a(V^*) > \frac{2}{k+1} \opt$.
    In the first case, note that the solution we computed for~$V^*$ consists of optimal solutions for~$G[V^*]$ and~${G[V\setminus V^*]}$ plus~$a(V^*)$.
    Since the first two summands are combined at most~$\opt$ and~${a(V^*) \leq \frac{2}{k+1} \opt}$, we indeed return a feedback arc set of size at most~${(1 + \frac{2}{k+1})\opt}$.
    If~$a(V^*) > \frac{2}{k+1} \opt$, then it holds for each~$V' \in \mathcal{V}$ that~${a(V') > \frac{2}{k+1} \opt}$.
    Moreover, since the graph~$G' = (V,A\setminus F)$ is acyclic, it has a topological ordering.
    Let~$V''$ be the set of the first~$\alpha n$ vertices in that topological ordering.
    Let~$d$ be the weight of a feedback arc set of minimum weight in~$G[V \setminus V'']$.
    Note that~$d \leq \frac{k-1}{k+1}\opt$.
    The solution we compute for~$V''$ has size~${(\opt - d) + (1+\frac{2}{k}) d  \leq (1+\frac{2(k-1)}{k(k+1)})\opt < (1+\frac{2}{k+1})\opt}$.
    Thus also in this case, we return a~$(1+\frac{2}{k+1})$-approximation.
    
    We conclude with analyzing the running time.
    Recall that~$\gamma = \frac{1}{\alpha}$.
    Computing~$\mathcal{V}$,~$a(V')$, and the optimal solution for~$G[V']$ for each~$V' \in \mathcal{V}$ takes~$O(\binom{n}{\alpha n} \poly(n) \log(\wmax))$ time.
    Solving the instance~$G[V \setminus V^*]$ optimally takes~$O(2^{(1-\alpha)n}  \poly(n) \log(\wmax))$ time.
    Approximating the solutions for all sets in~$\mathcal{V}$ takes $O(\binom{n}{\alpha n} \cdot (2-\delta_k)^{(1-\alpha)n} \cdot \log(\wmax))$ time in total.
    As shown in the proof for \cref{thm:fas}, it holds that~${O(\binom{n}{\alpha n} \cdot (2-\delta_k)^{(1-\alpha)n}) \subseteq O(2^{(1-\alpha)n})}$.
    Thus, the total running time of our algorithm is in~$O^*(2^{(1-\alpha)n})$ and in~$O((2-\delta)^n\log(\wmax))$ for any $0 < \delta < 2-2^{1-\alpha}$.
    This concludes the proof.
\end{proof}

\subsection{Cutwidth and Optimal Linear Arrangement}
In this final part, we generalize our algorithms for \textsc{Directed Cutwidth} and \dola{} to the weighted setting.
Since all generalizations follow the same simple structure, we only show the result for \textsc{Directed Cutwidth}.

\begin{proposition}
    \wcw{} admits a~$(2+\varepsilon)$-approximation in time $O(2^{\frac{\omega n}{3}} \log(\wmax) (\frac{1}{\varepsilon})^2k^2)$.
\end{proposition}

\begin{proof}
    The proof is similar to the proof of \cref{prop:cw}.
    Let~$(G=(V,A),w)$ be the input instance to \wcw.
    We find a $(1+\varepsilon)$-approximation for the instance~$(G,w,\lfloor \frac{n}{2} \rfloor)$ of \wdkmc{} using \cref{prop:wdkmc}.
    This takes~$O(2^{\frac{\omega n}{3}} \log(\wmax) (\frac{1}{\varepsilon})^2k^2)$ time.
    Let~$L$ be the set of vertices returned by \cref{prop:dkmc}, let~$R = V \setminus L$, and let~$A' \subseteq A$ be the set of arcs from~$R$ to~$L$.
    Note that~$\frac{w(A')}{1+\varepsilon}$ is by definition a lower bound for the weighted directed cutwidth of~$G$.
    Next, we can solve~$G[L]$ and~$G[R]$ optimally in~$O^*(2^{\frac{n}{2}} \log(\wmax))$ time~\cite{HK61}.
    Since adding a vertex and/or arc to a graph cannot decrease its cutwidth, we have that the weighted directed cutwidth of~$G[L]$ and~$G[R]$ are both at most the weighted directed cutwidth of~$G$.
    Adding the edges from~$L$ to~$R$ also does not increase the cutwidth.
    Hence, placing all vertices in~$L$ according to the computing solution for~$G[L]$ and then all vertices in~$R$ according to the solution for~$G[R]$ gives an ordering for the graph~$G'=(V,A\setminus A')$ that shows that the weighted directed cutwidth of~$G'$ is at most the weighted directed cutwidth of~$G$.
    Adding the set~$A'$ to~$G'$ can increase the cutwidth by at most~$w(A')$ and hence the computed ordering is a $(2+\varepsilon)$-approximation.
    The running time is in~$O(2^{\frac{\omega n}{3}} \log(\wmax) (\frac{1}{\varepsilon})^2k^2)$ as~$\frac{\omega}{3} > \frac{1}{2}$ and this concludes the proof.
\end{proof}

The generalizations for \ola{} and \dola{} are basically the same and are hence omitted.
We just mention that the slightly weaker running time~$(1-\frac{\alpha}{4})$ instead of~$(1-\frac{\alpha}{2})$ in the exponent comes from the fact that we only have a~$(1+\varepsilon)$-approximation for \wdkmc{} (where we set~$\varepsilon = \frac{\alpha}{2}$ and then compute the approximation for all~${\frac{\alpha}{4}n \leq k \leq (1-\frac{\alpha}{4})n}$ in order to compensate.
The correctness then follows from the fact that~$\frac{1+\frac{\alpha}{2}}{1-\frac{\alpha}{2}} \leq \frac{1}{1-\alpha}$ for all~$\alpha > 0$.

\begin{proposition}
        For each constant~$0 < \alpha < 1$, there is a~$(1+\frac{1}{1-\alpha})$-approximation for \dola{} that runs in~$O^*(2^{\frac{\omega n}{3}} + 2^{(1-\frac{\alpha}{4})n})$ time.
\end{proposition}

\begin{proposition}
        For each constant~$0 < \alpha < 1$, there is a~$(1+\frac{1}{2(1-\alpha)})$-approximation for \ola{} that runs in~$O^*(2^{\frac{\omega n}{3}} + 2^{(1-\frac{\alpha}{4})n})$ time.
\end{proposition}

\section{Conclusion}
\label{sec:concl}

In this work, we initiated the study of faster exponential-time approximation algorithms for vertex-ordering problems.
For \fas{} we designed an approximation scheme using a novel self-improving technique.
For other problems, we showed how to compute a constant-factor approximation.
We mention that the best known polynomial-time algorithms do not achieve a factor better  than~$O(\sqrt{\log(n)})$.

We conclude with a few open problems. 
One obvious direction is to design a faster exponential-time algorithm for \dkmc---exact or approximate---which would immediately yield faster approximation algorithms for almost all of the problems considered in this paper.
Next, there are a lot more (vertex-ordering) problems that are also interesting to study.
Examples include \textsc{Treewidth}, \textsc{Subset Feedback Vertex/Arc Set} (here also the restriction to tournament graphs is often studied), \textsc{Travelling Sales Person}, \textsc{Bandwidth}, or \textsc{Minimum Fill-In}.
Second, can our constant-factor approximations be improved to~$(1+\varepsilon)$ as in the case of \fas?
On a related note, we are not aware of any tool to show lower bounds for exponential-time approximations and most known hypotheses like SETH seem to weak to base such lower bounds on.
In particular, a reasonably sounding candidate one could call Gap-SETH (a combination of SETH and Gap-ETH) is not true as shown by Alman et al.~\cite{AlmanCW20}.
Hence, we ask the following question:
\begin{quote}
    What would be a sensible hypothesis to exclude~$(1+\varepsilon)$-approximations in~$O((2-\delta_\varepsilon)^n)$ time for some~$\delta_\varepsilon>0$ depending only on~$\varepsilon$ for different problems?
\end{quote}
Finally, all our algorithms use exponential space (here, the bottleneck is our exact algorithm for \dkmc{}, which first builds an exponential-size graph).
Can this be reduced to polynomial space while still beating the fastest known exact algorithms?

\bibliography{bib}

\begin{thebibliography}{38}
\providecommand{\natexlab}[1]{#1}
\providecommand{\url}[1]{\texttt{#1}}
\expandafter\ifx\csname urlstyle\endcsname\relax
  \providecommand{\doi}[1]{doi: #1}\else
  \providecommand{\doi}{doi: \begingroup \urlstyle{rm}\Url}\fi

\bibitem[Alman et~al.(2020)Alman, Chan, and Williams]{AlmanCW20}
Josh Alman, Timothy~M. Chan, and R.~Ryan Williams.
\newblock Faster deterministic and {L}as {V}egas algorithms for offline approximate nearest neighbors in high dimensions.
\newblock In \emph{Proceedings of the 31st {ACM-SIAM} Symposium on Discrete Algorithms (SODA)}, pages 637--649. {SIAM}, 2020.

\bibitem[Ambainis et~al.(2019)Ambainis, Balodis, Iraids, Kokainis, Prusis, and Vihrovs]{ABIKPV19}
Andris Ambainis, Kaspars Balodis, Janis Iraids, Martins Kokainis, Krisjanis Prusis, and Jevgenijs Vihrovs.
\newblock Quantum speedups for exponential-time dynamic programming algorithms.
\newblock In \emph{Proceedings of the 13th Annual {ACM-SIAM} Symposium on Discrete Algorithms ({SODA})}, pages 1783--1793. {SIAM}, 2019.

\bibitem[Austrin et~al.(2012)Austrin, Pitassi, and Wu]{APW12}
Per Austrin, Toniann Pitassi, and Yu~Wu.
\newblock Inapproximability of treewidth, one-shot pebbling, and related layout problems.
\newblock In \emph{Proceedings of the 15th International Workshop on Approximation, Randomization, and Combinatorial Optimization ({APPROX})}, pages 13--24. Springer, 2012.

\bibitem[Bansal et~al.(2019)Bansal, Chalermsook, Laekhanukit, Nanongkai, and Nederlof]{BCLNN19}
Nikhil Bansal, Parinya Chalermsook, Bundit Laekhanukit, Danupon Nanongkai, and Jesper Nederlof.
\newblock New tools and connections for exponential-time approximation.
\newblock \emph{Algorithmica}, 81\penalty0 (10):\penalty0 3993--4009, 2019.

\bibitem[Bansal et~al.(2024)Bansal, Katzelnick, and Schwartz]{BKS24}
Nikhil Bansal, Dor Katzelnick, and Roy Schwartz.
\newblock On approximating cutwidth and pathwidth.
\newblock In \emph{Procceedings of the 65th {IEEE} Annual Symposium on Foundations of Computer Science ({FOCS})}. {IEEE}, 2024.
\newblock Accepted for publication.

\bibitem[Bodlaender et~al.(2012)Bodlaender, Fomin, Koster, Kratsch, and Thilikos]{BFK+12}
Hans~L. Bodlaender, Fedor~V. Fomin, Arie M. C.~A. Koster, Dieter Kratsch, and Dimitrios~M. Thilikos.
\newblock A note on exact algorithms for vertex ordering problems on graphs.
\newblock \emph{Theory of Computing Systems}, 50\penalty0 (3):\penalty0 420--432, 2012.

\bibitem[Bonnet et~al.(2015)Bonnet, Escoffier, Paschos, and Tourniaire]{BEPT15}
Edouard Bonnet, Bruno Escoffier, Vangelis~Th. Paschos, and Emeric Tourniaire.
\newblock Multi-parameter analysis for local graph partitioning problems: Using greediness for parameterization.
\newblock \emph{Algorithmica}, 71\penalty0 (3):\penalty0 566--580, 2015.

\bibitem[Bourgeois et~al.(2011)Bourgeois, Escoffier, and Paschos]{BEP11}
Nicolas Bourgeois, Bruno Escoffier, and Vangelis~Th. Paschos.
\newblock Approximation of max independent set, min vertex cover and related problems by moderately exponential algorithms.
\newblock \emph{Discrete Applied Mathematics}, 159\penalty0 (17):\penalty0 1954--1970, 2011.

\bibitem[Chalermsook et~al.(2013)Chalermsook, Laekhanukit, and Nanongkai]{CLN13}
Parinya Chalermsook, Bundit Laekhanukit, and Danupon Nanongkai.
\newblock Independent set, induced matching, and pricing: Connections and tight (subexponential time) approximation hardnesses.
\newblock In \emph{Proceedings of the 54th Annual {IEEE} Symposium on Foundations of Computer Science ({FOCS})}, pages 370--379. {IEEE} Computer Society, 2013.

\bibitem[Charikar et~al.(2010)Charikar, Hajiaghayi, Karloff, and Rao]{CHKR10}
Moses Charikar, Mohammad~Taghi Hajiaghayi, Howard~J. Karloff, and Satish Rao.
\newblock $\ell_2^2$ spreading metrics for vertex ordering problems.
\newblock \emph{Algorithmica}, 56\penalty0 (4):\penalty0 577--604, 2010.

\bibitem[Cohen et~al.(2006)Cohen, Fomin, Heggernes, Kratsch, and Kucherov]{CFHKK06}
Johanne Cohen, Fedor~V. Fomin, Pinar Heggernes, Dieter Kratsch, and Gregory Kucherov.
\newblock Optimal linear arrangement of interval graphs.
\newblock In \emph{Proceedings of the 31st International Symposium on Mathematical Foundations of Computer Science ({MFCS})}, pages 267--279. Springer, 2006.

\bibitem[Cygan and Pilipczuk(2010)]{CP10}
Marek Cygan and Marcin Pilipczuk.
\newblock Exact and approximate bandwidth.
\newblock \emph{Theoretical Computer Science}, 411\penalty0 (40-42):\penalty0 3701--3713, 2010.

\bibitem[Cygan et~al.(2014)Cygan, Lokshtanov, Pilipczuk, Pilipczuk, and Saurabh]{CLPPS14}
Marek Cygan, Daniel Lokshtanov, Marcin Pilipczuk, Michal Pilipczuk, and Saket Saurabh.
\newblock On cutwidth parameterized by vertex cover.
\newblock \emph{Algorithmica}, 68\penalty0 (4):\penalty0 940--953, 2014.

\bibitem[Dinur and Safra(2005)]{DS05}
Irit Dinur and Samuel Safra.
\newblock On the hardness of approximating minimum vertex cover.
\newblock \emph{Annals of mathematics}, pages 439--485, 2005.

\bibitem[Esmer et~al.(2022)Esmer, Kulik, Marx, Neuen, and Sharma]{EsmerKMNS22ESA}
Baris~Can Esmer, Ariel Kulik, D{\'{a}}niel Marx, Daniel Neuen, and Roohani Sharma.
\newblock Faster exponential-time approximation algorithms using approximate monotone local search.
\newblock In \emph{Proceedings of the 30th Annual European Symposium on Algorithms ({ESA})}, pages 50:1--50:19. Schloss Dagstuhl - Leibniz-Zentrum f{\"{u}}r Informatik, 2022.

\bibitem[Esmer et~al.(2023)Esmer, Kulik, Marx, Neuen, and Sharma]{EsmerKMNS23IPEC}
Baris~Can Esmer, Ariel Kulik, D{\'{a}}niel Marx, Daniel Neuen, and Roohani Sharma.
\newblock Approximate monotone local search for weighted problems.
\newblock In \emph{Proceedings of the 18th International Symposium on Parameterized and Exact Computation ({IPEC})}, pages 17:1--17:23. Schloss Dagstuhl - Leibniz-Zentrum f{\"{u}}r Informatik, 2023.

\bibitem[Esmer et~al.(2024)Esmer, Kulik, Marx, Neuen, and Sharma]{EKMNS24}
Baris~Can Esmer, Ariel Kulik, D{\'{a}}niel Marx, Daniel Neuen, and Roohani Sharma.
\newblock Optimally repurposing existing algorithms to obtain exponential-time approximations.
\newblock In \emph{Proceedings of the 2024 {ACM-SIAM} Symposium on Discrete Algorithms ({SODA})}, pages 314--345. {SIAM}, 2024.

\bibitem[Even et~al.(1998)Even, Naor, Schieber, and Sudan]{ENSS98}
Guy Even, Joseph Naor, Baruch Schieber, and Madhu Sudan.
\newblock Approximating minimum feedback sets and multicuts in directed graphs.
\newblock \emph{Algorithmica}, 20\penalty0 (2):\penalty0 151--174, 1998.

\bibitem[Feige and Lee(2007)]{FL07}
Uriel Feige and James~R. Lee.
\newblock An improved approximation ratio for the minimum linear arrangement problem.
\newblock \emph{Information Processing Letters}, 101\penalty0 (1):\penalty0 26--29, 2007.

\bibitem[Feige et~al.(2008)Feige, Hajiaghayi, and Lee]{FHL08}
Uriel Feige, MohammadTaghi Hajiaghayi, and James~R. Lee.
\newblock Improved approximation algorithms for minimum weight vertex separators.
\newblock \emph{{SIAM} Journal on Computing}, 38\penalty0 (2):\penalty0 629--657, 2008.

\bibitem[Fomin and Kratsch(2010)]{Fomin:2010mo}
Fedor~V. Fomin and Dieter Kratsch.
\newblock \emph{Exact Exponential Algorithms}.
\newblock Springer, 2010.

\bibitem[Fomin et~al.(2016)Fomin, Gaspers, Lokshtanov, and Saurabh]{FominGLS16}
Fedor~V. Fomin, Serge Gaspers, Daniel Lokshtanov, and Saket Saurabh.
\newblock Exact algorithms via monotone local search.
\newblock In \emph{Proceedings of the 48th Annual {ACM} {SIGACT} Symposium on Theory of Computing ({STOC})}, pages 764--775. {ACM}, 2016.

\bibitem[Held and Karp(1961)]{HK61}
Michael Held and Richard~M. Karp.
\newblock A dynamic programming approach to sequencing problems.
\newblock In \emph{Proceedings of the 16th {ACM} national meeting}, page~71. {ACM}, 1961.

\bibitem[Husfeldt et~al.(2010)Husfeldt, Kratsch, Paturi, and Sorkin]{HKPS10}
Thore Husfeldt, Dieter Kratsch, Ramamohan Paturi, and Gregory~B. Sorkin, editors.
\newblock \emph{Exact Complexity of NP-hard Problems}, volume 10441 of \emph{Dagstuhl Seminar Proceedings}, 2010. Schloss Dagstuhl - Leibniz-Zentrum f{\"{u}}r Informatik, Germany.

\bibitem[Husfeldt et~al.(2013)Husfeldt, Paturi, Sorkin, and Williams]{HPSW13}
Thore Husfeldt, Ramamohan Paturi, Gregory~B. Sorkin, and Ryan Williams, editors.
\newblock \emph{Exponential Algorithms: Algorithms and Complexity Beyond Polynomial Time}, volume 13331 of \emph{Dagstuhl Seminar Proceedings}, 2013. Schloss Dagstuhl - Leibniz-Zentrum f{\"{u}}r Informatik, Germany.

\bibitem[Inamdar et~al.(2024)Inamdar, Kundu, Parviainen, Ramanujan, and Saurabh]{Tanmay24}
Tanmay Inamdar, Madhumita Kundu, Pekka Parviainen, M.~S. Ramanujan, and Saket Saurabh.
\newblock Exponential-time approximation schemes via compression.
\newblock In \emph{Proceedings of the 15th Conference on Innovations in Theoretical Computer Science (ITCS)}, pages 64:1--64:22. Schloss Dagstuhl - Leibniz-Zentrum f{\"u}r Informatik, 2024.

\bibitem[Itai and Rodeh(1978)]{IR78}
Alon Itai and Michael Rodeh.
\newblock Finding a minimum circuit in a graph.
\newblock \emph{{SIAM} Journal on Computing}, 7\penalty0 (4):\penalty0 413--423, 1978.

\bibitem[Jain et~al.(2019)Jain, Kanesh, Lochet, Saurabh, and Sharma]{JKLSS19}
Pallavi Jain, Lawqueen Kanesh, William Lochet, Saket Saurabh, and Roohani Sharma.
\newblock Exact and approximate digraph bandwidth.
\newblock In \emph{Proceedings of the 39th {IARCS} Annual Conference on Foundations of Software Technology and Theoretical Computer Science ({FSTTCS})}, pages 18:1--18:15. Schloss Dagstuhl - Leibniz-Zentrum f{\"{u}}r Informatik, 2019.

\bibitem[Kann(1992)]{Kann92}
Viggo Kann.
\newblock \emph{On the approximability of {NP}-complete optimization problems}.
\newblock PhD thesis, Royal Institute of Technology Stockholm, 1992.

\bibitem[Kitsunai et~al.(2016)Kitsunai, Kobayashi, Komuro, Tamaki, and Tano]{KKKTT16}
Kenta Kitsunai, Yasuaki Kobayashi, Keita Komuro, Hisao Tamaki, and Toshihiro Tano.
\newblock Computing directed pathwidth in {$O(1.89^n)$} time.
\newblock \emph{Algorithmica}, 75\penalty0 (1):\penalty0 138--157, 2016.

\bibitem[Lokshtanov et~al.(2023)Lokshtanov, Saurabh, and Surianarayanan]{Lokshtanov0S23}
Daniel Lokshtanov, Saket Saurabh, and Vaishali Surianarayanan.
\newblock Breaking the all subsets barrier for min k-cut.
\newblock In \emph{Procceedings of the 50th International Colloquium on Automata, Languages, and Programming ({ICALP})}, pages 90:1--90:19. Schloss Dagstuhl - Leibniz-Zentrum f{\"{u}}r Informatik, 2023.

\bibitem[Mihai and Todinca(2009)]{MT09}
Rodica Mihai and Ioan Todinca.
\newblock Pathwidth is {NP}-hard for weighted trees.
\newblock In \emph{Proceedings of the 3rd International Workshop on Frontiers in Algorithmics ({FAW})}, pages 181--195. Springer, 2009.

\bibitem[Nederlof(2020)]{Nederlof20}
Jesper Nederlof.
\newblock Bipartite {TSP} in $o(1.9999^n)$ time, assuming quadratic time matrix multiplication.
\newblock In \emph{Proceedings of the 52nd Annual {ACM} {SIGACT} Symposium on Theory of Computing ({STOC})}, pages 40--53. {ACM}, 2020.

\bibitem[Nederlof et~al.(2023)Nederlof, Pawlewicz, Swennenhuis, and Wegrzycki]{NPSW23}
Jesper Nederlof, Jakub Pawlewicz, C{\'{e}}line M.~F. Swennenhuis, and Karol Wegrzycki.
\newblock A faster exponential time algorithm for bin packing with a constant number of bins via additive combinatorics.
\newblock \emph{{SIAM} Journal of Comput.}, 52\penalty0 (6):\penalty0 1369--1412, 2023.

\bibitem[Robbins(1955)]{Rob55}
Herbert Robbins.
\newblock A remark on {S}tirling's formula.
\newblock \emph{The American Mathematical Monthly}, 62\penalty0 (1):\penalty0 26--29, 1955.

\bibitem[Williams(2004)]{Williams04}
Ryan Williams.
\newblock A new algorithm for optimal constraint satisfaction and its implications.
\newblock In \emph{Proceedings of the 31st International Colloquium on Automata, Languages and Programming ({ICALP})}, pages 1227--1237. Springer, 2004.

\bibitem[Zamir(2021)]{Zamir21}
Or~Zamir.
\newblock Breaking the $2^n$ barrier for 5-coloring and 6-coloring.
\newblock In \emph{Procceedings of the 48th International Colloquium on Automata, Languages, and Programming ({ICALP})}, pages 113:1--113:20. Schloss Dagstuhl - Leibniz-Zentrum f{\"{u}}r Informatik, 2021.

\bibitem[Zwick(2002)]{Zwick02}
Uri Zwick.
\newblock All pairs shortest paths using bridging sets and rectangular matrix multiplication.
\newblock \emph{Journal of the {ACM}}, 49\penalty0 (3):\penalty0 289--317, 2002.

\end{thebibliography}

\end{document}